\documentclass[12pt]{amsart}

\usepackage{amsmath,amssymb,enumerate,mathtools}
\newtheorem{theorem}{Theorem}[section]
\newtheorem{claim}{Claim}[theorem]
\newtheorem{lemma}[theorem]{Lemma}

\newtheorem{conjecture}[theorem]{Conjecture}

\theoremstyle{definition}

\newcommand{\bR}{\mathbb R}
\newcommand{\bZ}{\mathbb Z}

\newcommand{\cC}{\mathcal{C}}

\newcommand{\cG}{\mathcal{G}}

\newcommand{\cZ}{\mathcal{Z}}

\newcommand{\eps}{\varepsilon}
\newcommand{\half}{\tfrac{1}{2}}

\DeclareMathOperator{\GF}{GF}

\newcommand{\prob}{\mathbf{P}}
\newcommand{\ev}{\mathbf{E}}

\newcommand{\norm}[1]{\left\lVert #1 \right\rVert}

\newcommand{\allones}{\overline{\mathbf{1}}}
\begin{document}
\sloppy
\allowdisplaybreaks
\title[The MLD threshold for cycle codes ]{The maximum-likelihood decoding threshold for cycle codes of graphs}

\author[Nelson]{Peter Nelson}
\author[Van Zwam]{Stefan H.M. van Zwam}
\thanks{ This research was partially supported by a grant from the
Office of Naval Research [N00014-12-1-0031] and by grants from the National Science Foundation Division of Mathematical Sciences [1501985, 1500343]}
%\subjclass{05B35}
%\keywords{matroids, growth rates}
%\date{\today}
\begin{abstract}
	For a class $\cC$ of binary linear codes, we write $\theta_{\cC}\colon (0,1) \to [0,\tfrac{1}{2}]$ for the \emph{maximum-likelihood decoding threshold function} of $\cC$, the function whose value at $R \in (0,1)$ is the largest bit-error rate $p$ that codes in $\cC$ can tolerate with a negligible probability of maximum-likelihood decoding error across a binary symmetric channel. We show that, if $\cC$ is the class of cycle codes of graphs, then $\theta_{\cC}(R) \le \tfrac{(1-\sqrt{R})^2}{2(1+R)}$ for each $R$, and show that equality holds only when $R$ is asymptotically achieved by the cycle codes of regular graphs. 
\end{abstract}

\maketitle

\section{Introduction}

	For a class $\cC$ of binary linear codes and for some rate $R \in (0,1)$, we consider the \emph{maximum-likelihood decoding threshold} $\theta_{\cC}(R)$ for $\cC$ at $R$. This is the unique $\theta \in [0,\half]$ such that
	\begin{itemize}
		\item for each $p \in (0,\theta)$ and all $\eps > 0$, given a binary symmetric channel of bit-error rate $p$, there exists a code $C \in \cC$ of rate at least $R$ such that the probability of a error in maximum-likelihood decoding on $C$ is at most $\eps$, and
		\item for each $p \in (\theta,\half)$ there exists $\eps > 0$ such that, given a binary symmetric channel of bit-error rate $p$, for each code $C \in \cC$ of rate at least $R$ the probability of an error in maximum-likelihood decoding on $C$ is at least $\eps$. 
	\end{itemize} 
	The function $\theta_{\cC}(R)$ is the \emph{threshold function} for $\cC$; it essentially measures the maximum bit-error rate that can be `tolerated' by rate-$R$ codes in $\cC$ with vanishing probability of a decoding error. Our main result proves an upper bound on this function for the class $\cG$ of cycle codes of graphs: 
	\begin{theorem}\label{mainintro}
		If $\cG$ is the class of cycle codes of graphs and $R \in (0,1)$, then $\theta_{\cG}(R) \le \tfrac{(1-\sqrt{R})^2}{2(1+R)}$. If equality holds, then $R = 1 - \tfrac{2}{d}$ for some $d \in \bZ$. 
	\end{theorem}
	
	This generalises a result of Decreusefond and Z\'{e}mor [\ref{dz}], who proved the same upper bound for the class of cycle codes of regular graphs. Our proof follows theirs conceptually, although our exposition and notation are somewhat different. The proof in [\ref{dz}] implicitly involves a problem of enumerating `non-backtracking' walks that is trivial for regular graphs but not in general; much of the original material in our proof is related to this difficulty. 
	
	When $R = 1 - \tfrac{2}{d}$ for some $d \in \bZ$ (that is, when the cycle codes of large $d$-regular graphs have rate close to $R$) our theorem does not improve the bound $\theta_{\cG}(R) \le \tfrac{(1-\sqrt{R})^2}{2(1+R)}$. In this case, however, the bound is known to be best-possible; Z\'{e}mor and Tillich [\ref{zt97}] showed, when $d-1$ is one of various prime powers, that certain families of $d$-regular Ramanujan graphs have cycle codes attaining this threshold (that is, can tolerate a bit-error rate of $p$ for any $p < \tfrac{(1-\sqrt{R})^2}{2(1+R)}$), and later random constructions due to Alon and Bachmat [\ref{ab06}] can be demonstrated to give the same result for all $d \ge 3$. Combining these constructions with Theorem~\ref{mainintro}, we have the following:
	
	\begin{theorem}\label{intval}
		If $\cG$ is the class of cycle codes of graphs, and $R = 1 - \tfrac{2}{d}$ for some integer $d \ge 3$, then $\theta_{\cG}(R) = \tfrac{(1-\sqrt{R})^2}{2(1+R)}$.
	\end{theorem}
	Theorem~\ref{mainintro} implies that this equality holds for no other $R \in (0,1)$; this can be interpreted as a statement that the cycle codes of regular graphs are `best' among all cycle codes. 
	
	Theorem~\ref{mainintro} will be derived as a consequence of a stronger upper bound for $\theta_{\cG}$, given in Theorem~\ref{maintechnical}. While the bound in Theorem~\ref{maintechnical} is highly technical in its statement, we believe (Conjecture~\ref{mainconjecture}) that it is in fact the correct upper bound. 
	
	%a consequence of the following stronger result in which we obtain a much more technical upper bound, which we state here omitting a crucial definition.\begin{theorem}\label{mainintrotech}
	 %	If $\cG$ is the class of cycle codes and $R \in (0,1)$, then $\theta_{\cG}(R) \le \tfrac{1}{2}\left(1-\sqrt{1+\tfrac{2}{\lambda}}\right)$, where $\lambda = \lambda_*\left(\tfrac{2}{1-R}\right)$.
	%\end{theorem}	
	%  The function $\lambda_*(\cdot)$, whose definition appears at the end of Section~\ref{walksection}, is defined in terms of eigenvalues of the adjacency matrices of with directed graphs. Although Theorem~\ref{mainintrotech} is technical, we state it here since we believe it provides the true threshold function.
	  
	 % \begin{conjecture}
	  %	The inequality in Theorem~\ref{mainintrotech} holds with equality for all $R \in (0,1)$.
	%\end{conjecture}

	\subsection*{Minor-Closed Classes}
 The main result of [\ref{nvz14}] shows that the failure of the cycle codes to be `asymptotically good' extends to every \emph{proper minor-closed} subclass of binary codes; that is, every proper subclass that is closed under puncturing and shortening. The proof uses a deep result in matroid structure theory due to Geelen, Gerards and Whittle [\ref{ggw}] that states, roughly, that the `highly connected' members of any such class of codes are close to being either cycle codes or their duals.  
	
	We believe that this paradigm that the members of any minor-closed subclass of binary codes are `nearly' cycle or cocycle codes will also apply to the threshold function. We predict that the threshold function $\theta_{\cG}(R)$ for any minor-closed class agrees with that of either the class $\cG$ of cycle codes or the class $\cG^*$ of cocycle codes. It is easily shown (see [\ref{ggw}]) that $\theta_{\cG^*}(R) = 0$ for all $R \in (0,1)$.  Geelen, Gerards and Whittle [\ref{ggw}] made the following striking conjecture:
	\begin{conjecture}
		Let $\cC$ be a proper subclass of the binary linear codes that is closed under puncturing and shortening. Either
		\begin{itemize}
			\item $\cG \subseteq \cC$ and $\theta_{\cC} = \theta_{\cG}$, or 
			\item $\theta_{\cC} = 0$. 
		\end{itemize}
	\end{conjecture}
	In other words, the presence or absence of the class of cycle codes should be all that determines the threshold function for any minor-closed class. Proving this conjecture would likely require a combination of the matroidal techniques in [\ref{nvz14}] and the algebraic and probabilistic ideas in this paper.
	
\section{Preliminaries} 

	We give some basic definitions in coding theory that, together with the definition of threshold function in the introduction, are all that are required for this paper; a more comprehensive reference is found in [\ref{ms77}]. We also use some standard graph theory terminology from [\ref{diestel}] and [\ref{gr}]. 
	
	For integers $n \ge k \ge 0$, a \emph{binary linear $[n,k]$-code} is a $k$-dimensional subspace $C$ of some $n$-dimensional vector space $V$ over $\GF(2)$. We call the elements of $C$ \emph{codewords}. The \emph{rate} of $C$ is the ratio $R = \tfrac{k}{n}$.
	
	\subsection{Cycle codes}This paper is concerned solely with the cycle codes of graphs. For a finite  graph $G = (V,E)$, the \emph{cycle code} of $G$ is the subspace of $\GF(2)^{E}$ whose elements are exactly the characteristic vectors of cycles of $G$ (that is, edge-disjoint unions of circuits of $G$, or equivalently edge-sets of even subgraphs of $G$). We write $\cG$ for the class of all such codes; it is well-known that every cycle code is the cycle code of a connected graph. 
	
	If $G$ is connected, then its cycle code $C$ is a binary linear $[n,k]$-code, where $n = |E|$ and $k = |E| - |V| + 1$, giving $R = 1 - \tfrac{|V|}{|E|} + \tfrac{1}{|E|}$. The ratio $\tfrac{|V|}{|E|}$ is exactly $\tfrac{2}{\mu(G)}$, where $\mu(G)$ denotes the average degree of $G$; we adopt this notation $\mu(G)$ throughout the paper. The above formula implies that a large connected graph $G$ has a cycle code of rate $R \approx 1 - \tfrac{2}{\mu(G)}$. A simple `error-tolerance' parameter of $C$  is the minimum Hamming distance $d$ between two codewords of $C$; this is equal to the \emph{girth} of $G$ (the length of a shortest circuit of $G$) -- we will write $d(G)$ for the girth of a graph $G$. 
	
	\subsection{Maximum-likelihood decoding} Suppose that some codeword $c$ of a linear $[n,k]$-code $C \subseteq V$ is transmitted across a binary symmetric channel with bit-error rate $p \in (0,\half)$, giving some $x \in V$ obtained by switching the value of each entry of $c$ independently with probability $p$. \emph{Maximum-likelihood decoding} (abbreviated ML-decoding) is the process where, given $x$, we attempt to recover $c$ by choosing the codeword $c' \in C$ with the highest probability to have been sent, given that $x$ has been received. If this choice is ambiguous (that is, if this maximum is not unique) or gives  an incorrect answer (that is, if $c' \ne c$), then we say a \emph{decoding error} has been made; this occurs with some probability depending on $p$ and $C$ but, by linearity, not on the particular codeword $c$. In this particular setting of a constant bit-error probability $p < \half$ that behaves independently on each bit, ML-decoding is equivalent to \emph{nearest-neighbour} decoding, where $c'$ is simply chosen to be the closest codeword to $x$ in Hamming distance. We remark that our definition of ML-decoding deviates slightly from the standard one, in which a decoding error is also avoided with nonzero probability in the case of an ambiguous choice. This difference will not affect the asymptotic analysis with which we are concerned. 
		
	ML-decoding is hard for general binary codes [\ref{bmt78}], but an attractive property of cycle codes of graphs (and an important motivating factor for this paper) is that ML-decoding can be implemented efficiently for cycle codes using standard techniques in combinatorial optimization (see [\ref{nh81}]). This is the case because the probability of a decoding error can be understood purely graphically: if $C$ is the cycle code of a graph $G = (V,E)$ and codewords of $C$ are transmitted across a channel of bit-error rate $p \in (0, \half)$, then the probability of an ML-decoding error is exactly the probability, given a set $X \subseteq E$ formed by choosing each edge uniformly at random with probability $p$, that $X$ contains at least half of the edges of some circuit of $G$. Thus, to prove our main theorem, we study random subsets of edges of a graph. From this point on, given a set $E$ and some $p \in [0,1]$, we refer to a random set $X \subseteq E$ formed by including each element of $E$ independently at random with probability $p$ as a \emph{$p$-random subset of $E$}.  
		
\section{Non-backtracking walks}\label{walksection}

A \emph{non-backtracking walk} of length $\ell$ in a graph $G$ is a walk $(v_0,v_1,\dotsc,v_\ell)$ of $G$ so that $v_{i+1} \ne v_{i-1}$ for all $i \in \{1,\dotsc, \ell-1\}$. In all nontrivial cases, the number of such walks grows roughly exponentially in $\ell$; in this section we estimate the base of this exponent, mostly following ([\ref{alon}], Theorem 1).  

Let $G = (V,E)$ be a simple connected graph of minimum degree at least $2$. 
Let $\bar{E} = \{(u,v) \in V^2: u \sim_G v\}$ be the $2|E|$-element set of arcs of $G$. Let $B = B(G) \in \{0,1\}^{\bar{E} \times \bar{E}}$ be the matrix so that $B_{(u,v),(u',v')} = 1$ if and only if $u' = v$ and $u \ne v'$. It is easy to see that 
\begin{enumerate}
	\item\label{scd} $B$ is the adjacency matrix of a strongly connected digraph (essentially the `line digraph' of $G$), and 
	\item\label{nbw} For each integer $\ell \ge 1$, the entry $(B^{\ell})_{e,f}$ is  the number of non-backtracking walks of length $\ell+1$ in $G$ with first arc $(v_0,v_1) = e$ and last arc $(v_\ell,v_{\ell+1})= f$.
\end{enumerate}

 By (\ref{scd}) and the Perron-Frobenius theorem (see [\ref{gr}], section 8.8), there is a positive real eigenvalue $\lambda_*$ of $B$ and an associated positive real eigenvector $w_*$, so that $|\lambda_*| \ge |\lambda|$ for every eigenvalue $\lambda$ of $B$. Furthermore, by Gelfand's formula [\ref{gelfand}] we have $\lambda_* = \lim_{n \to \infty} \norm{B^n}^{1/n}$, where $\norm{B^n}$ denotes the sum of the absolute values of the entries of $B^n$. By (\ref{nbw}), the parameter $\lambda_* = \lambda_*(B(G))$ thus governs the growth of non-backtracking walks in $G$.
 
 Note that %$\bfj^T B^{\ell} \bfj$ is the total number of non-backtracking walks of length $\ell+1$ in $G$, and that 
 $B^\ell$ has only nonnegative entries, so $\norm{B^\ell} = \allones^T B^{\ell} \allones$. Let $\mu = \mu(G) = \tfrac{1}{n}|\bar{E}|$ denote the average degree of $G$. The proof of Theorem 1 of [\ref{alon}] contains the following:
 
 \begin{lemma} Let $G$ be a connected graph of minimum degree at least $2$ and let $B = B(G)$. Then $\allones^TB^{\ell} \allones \ge (n\mu) \Lambda^{\ell}$, where 
 \[\Lambda = \Lambda(G) = \prod_{v \in V}(d_G(v)-1)^{d_G(v)/(n \mu)}.\]
 \end{lemma}
  It follows in turn from this lemma that $\norm{B^\ell}^{1/\ell} \ge \Lambda(G)$, so $\lambda_*(B(G)) \ge \Lambda(G)$.  As observed in [\ref{alon}], the log-convexity of the function $(x-1)^x$ (for $x > 1$) implies that $\Lambda(G) \ge \mu(G)-1$. For each $x \in \bR$, let $\eta(x) = \min(x - \lfloor x \rfloor, \lceil x \rceil - x)$ denote the distance from $x$ to the nearest integer. The following lemma, which is proved by slightly improving the bound $\Lambda(G) \ge \mu(G)-1$ when $\mu(G)$ is not an integer, is an unilluminating exercise in calculus.
	
	\begin{lemma}\label{boundlambda}
		Let $\mu_0 \in \bR$ satisfy $\mu_0 \ge 2$ and let $G$ be a connected graph with minimum degree at least $2$ and average degree at least $\mu_0$. Then $\lambda_*(B(G)) \ge \mu_0-1 + \tfrac{\eta(\mu_0)^3}{8\mu_0^3}$. 
	\end{lemma}
	\begin{proof}
		Let $n = |V(G)|$, let $d_1, \dotsc, d_n$ be the degrees of the vertices of $G$, and let $\mu = \tfrac{1}{n}\sum_{i=1}^nd_i \ge \mu_0$ be the average degree of $G$. Let $\eta = \eta(\mu)$; note that $\mu \ge 2 + \eta$. Define $g\colon (1,\infty) \to \bR$ by $g(x) = x \ln (x-1)$; observe that $g'(x) = \tfrac{x}{x-1} + \ln(x-1)$ and $g''(x) = \tfrac{x-2}{(x-1)^2}$. We have 
		$\ln (\Lambda(G)) = \tfrac{1}{n \mu} \sum_{i = 1}^n g(d_i)$; for each $i$, Taylor's theorem gives 
		\[g(d_i) = g(\mu) + g'(\mu)(d_i-\mu) + \tfrac{1}{2}g''(\xi_i)(d_i-\mu)^2\]
		for some $\xi_i$ between $d_i$ and $\mu_0$. We now estimate the `error' terms.
		\begin{claim}
			$\tfrac{1}{2}g''(\xi_i)(d_i-\mu)^2 \ge \tfrac{\eta^3}{8\mu^2}$ for each $i$. 
		\end{claim}
		\begin{proof}[Proof of claim:]
		First suppose that $d_i = 2$. Then $g(d_i) = 0$, so 
		\begin{align*} \tfrac{1}{2}g''(\xi_i)(2-\mu)^2 &= -g(\mu) - g'(\mu)(2-\mu)  \\
		&= (\mu-2)\left(\tfrac{\mu}{\mu-1} + \ln(\mu-1)\right) - \mu \ln(\mu-1)\\
		&= \tfrac{\mu(\mu-2)}{\mu-1} - 2 \ln(\mu-1).
		\end{align*}
		Note that the above expression is equal to $1.174\dotsc > 1$ for $\mu = \tfrac{7}{3}$, and is increasing in $\mu$ for $\mu \in (2,\infty)$. If $\mu \ge \tfrac{7}{3}$ then we therefore have $\tfrac{1}{2}g''(\xi_i)(2-\mu)^2 > 1$. If $\mu < \tfrac{7}{3}$ then $\mu = 2 + \eta$ and $\eta < \tfrac{1}{3}$, so 
		\begin{align*}
			\tfrac{\mu(\mu-2)}{\mu-1} - 2 \ln(\mu-1) &= \tfrac{\eta(2 + \eta)}{1 + \eta} - 2 \ln(1 + \eta) \\
			&\ge \tfrac{\eta(2 + \eta)}{1 + \eta} - 2(\eta - \tfrac{1}{2} \eta^2 + \tfrac{1}{3}\eta^3)\\
			&= \tfrac{\eta^3}{3(1+\eta)}(1-2\eta)\\
			&> \tfrac{1}{12}\eta^3,
		\end{align*}
		where the last inequality uses $\eta < \tfrac{1}{3}$. Therefore if $d_i = 2$ we have $\tfrac{1}{2}g''(\xi_i)(d_i-\mu)^2 \ge \min(1,\tfrac{1}{12}\eta^3) = \tfrac{1}{12}\eta^3 >  \tfrac{\eta^3}{8\mu^2}$. 
		
		Suppose that $d_i \ge 3$. Since $\xi_i$ is between $\mu$ and $d_i$, we have $\xi_i \ge \min(d_i,\mu) \ge \min(3,\mu)$, so $\xi_i - 2 \ge \eta$. Therefore $g''(\xi_i) \ge \tfrac{\eta}{(\xi_i-1)^2} > \tfrac{\eta}{\xi_i^2}$. Thus, using $\xi_i \le \max(\mu,d_i)$, we have
		\[\tfrac{1}{2}g''(\xi_i)(d_i-\mu)^2 \ge \frac{\eta(d_i-\mu)^2}{2\xi_i^2} \ge \frac{\eta(d_i-\mu)^2}{2\max(\mu,d_i)^2}.\]
		It is easy to show, since $d_i \in \bZ$, that $\left|\tfrac{d_i-\mu}{\max(\mu,d_i)}\right| \ge \tfrac{\eta}{\mu + \eta} \ge \tfrac{\eta}{2\mu}$, so $\tfrac{1}{2}g''(\xi_i)(d_i-\mu)^2 \ge \tfrac{\eta^3}{8 \mu^2}$ and the claim follows. 
		\end{proof}
		
		Using the claim, we have 
		\begin{align*}
			\ln(\Lambda(G)) &= \frac{1}{n\mu} \sum_{i = 1}^n g(d_i)\\
			&= \frac{1}{n\mu}\sum_{i=1}^n \left(g(\mu) + g'(\mu)(d_i-\mu) + \tfrac{1}{2}g''(\xi_i)(d_i-\mu)^2\right)\\
			&= \frac{1}{n\mu}\left(n g(\mu) + \sum_{i=1}^n\tfrac{1}{2}g''(\xi_i)(d_i-\mu)^2\right)\\
			&\ge \ln(\mu-1) + \frac{1}{n\mu}\left(\frac{n\eta^3}{8\mu^2}\right)\\
			&= \ln(\mu-1) + \frac{\eta^3}{8\mu^3}.
		\end{align*}
		So $\Lambda(G) \ge (\mu-1)\exp\left(\tfrac{\eta^3}{8\mu^3}\right) \ge \mu-1 + (\mu-1)\left(\frac{\eta^3}{8\mu^3}\right) \ge \mu-1 + \frac{\eta^3}{8\mu^3}$. One easily checks that the function $h(y) = y-1 + \tfrac{\eta(y)}{8y^3}$ is strictly increasing on $(2,\infty)$; since $\mu \ge \mu_0$ and $\lambda_*(B(G)) \ge \Lambda(G)$, it follows that $\lambda_*(B(G)) \ge \mu_0-1 + \tfrac{\eta(\mu_0)^3}{8\mu_0^3}$, as required.  
	\end{proof} 
	
	For each $\mu \ge 2$, let $\cG_{\mu}$ denote the class of connected graphs with average degree at least $\mu$ and minimum degree at least $2$. For every integer $n \ge \mu+1$, let \[\lambda_*(\mu;n) = \inf\{\lambda_*(B(G))\colon G \in \cG_\mu, |V(G)| = n\},\] noting that this infimum is finite since $K_n \in \cG_{\mu}$ for all $n \ge \mu+1$. Define $\lambda_*\colon [2,\infty) \to \bR$ by $\lambda_*(\mu) = \liminf_{n \to \infty} \lambda_*(\mu;n)$. The following is immediate from Lemma~\ref{boundlambda}.
	
	\begin{lemma}\label{lambdastar}
		 $\lambda_*(\mu) \ge \mu-1$. If equality holds, then $\mu \in \bZ$. 
	\end{lemma}
	
	Having defined the function $\lambda_*$, we can now state the more technical main theorem from which Theorem~\ref{mainintro} will easily follow. 
	
	\begin{theorem}\label{maintechnical}
	 	If $\cG$ is the class of cycle codes of graphs and $R \in (0,1)$, then $\theta_{\cG}(R) \le \tfrac{1}{2}\left(1-\sqrt{1-\tfrac{1}{\lambda^2}}\right)$, where $\lambda = \lambda_*\left(\tfrac{2}{1-R}\right)$.
	\end{theorem}	
	
	As mentioned, we believe the above bound is the true value for $\theta_{\cG}$.
	\begin{conjecture}\label{mainconjecture}
		The bound in Theorem~\ref{maintechnical} holds with equality for all $R \in (0,1)$. 
	\end{conjecture}
	 By Theorem~\ref{intval}, this conjecture holds when $R = 1 - \tfrac{2}{d}$ for $d \in \bZ$.

\section{Covering trees}

A \emph{locally finite, infinite rooted tree} (hereafter just a \emph{tree}) is a connected acyclic infinite graph $\Gamma$ of finite maximum degree together with a particular vertex $r$ called the \emph{root}. Adopting some notation of [\ref{dz}] and [\ref{lyons}], for $x \in V(\Gamma)$ we write $|x|$ for the distance of $x$ from $r$, and we write $x \preceq y$ if $x$ is on the path from $r$ to $y$. We write $x \wedge y$ for the \emph{join} of $x$ and $y$, the vertex of largest distance from $r$ that is on both the path from $r$ to $x$ and the path from $r$ to $y$. 

The trees we are interested in are `covering trees' for finite graphs. Let $G = (V,E)$ be a finite graph of minimum degree at least $2$ and let $e = (u,v)$ be an arc of $G$. The \emph{covering tree of $G$ rooted at $e$} is the tree $\Gamma = \Gamma_e(G)$ where the root is the length-zero walk $(u)$ of $G$, the other vertices are the non-backtracking walks of $G$ with first arc $e$ and the children of each walk $(u,v,v_2,\dotsc,v_\ell)$ of length $\ell$ are its extensions $(u,v,v_2,\dotsc,v_\ell,v_{\ell+1})$ to nonbacktracking walks of length $\ell+1$ (ie. where $v_{\ell+1}$ is adjacent to $v_{\ell}$ in $G$ and is not equal to $v_{\ell-1}$). Note that the number of vertices of $\Gamma_e(G)$ at distance $\ell$ from the root is the total number of length-$\ell$ non-backtracking walks of $G$ with first arc $e$, which is exactly the sum of the entries of the $e$-column of $B(G)^{\ell-1}$.  

There is a natural homomorphism that associates each walk with its final vertex; if $G$ has large girth, this map preserves much of the local structure of $G$. To analyse the ubiquity of cycles in a random sample of edges of $G$, we follow [\ref{dz}] and study a problem of `fractional percolation' on covering trees, bounding the probability that, given a $p$-random subset of $E(\Gamma_{e}(G))$, there is a long path starting at $r$ that is, in a certain sense, dense with edges in the subset. 

Let $\Gamma$ be such a tree, and let $\alpha \in (0,1)$. Given $X \subseteq E(\Gamma)$, we say that a finite path $(v_0,v_1, \dotsc, v_n)$ of $\Gamma$ is \emph{$\alpha$-adapted} with respect to $X$ if, for each $i \in \{1, \dotsc, n\}$, the subpath $(v_0,\dotsc,v_i)$ contains at least $\alpha i$ edges of $X$. If $t_1,t_2, \dotsc, $ is a sequence of positive integers and $T_n = \sum_{i=1}^n t_i$ is its sequence of partial sums (with $T_0 = 0$), then  we say that a path $(x_0,x_1, \dotsc, x_n)$ of $\Gamma$ is \emph{$(\alpha,t)$-adapted} with respect to $X$ if for each $i \in \bZ_{>0}$ for which $T_{i+1} < n$, the path $(x_{T_{i}},x_{T_{i}+1}, \dotsc, x_{T_{i+1}-1})$ is $\alpha$-adapted, and also the path $(x_{T_j},x_{T_j+1}, \dotsc, x_n)$ is $\alpha$-adapted, where $j$ is minimal so that $T_{j+1} > n$.  Note that any initial subpath of an $(\alpha,t)$-adapted path is $(\alpha,t)$-adapted. %A vertex $x$ of $\Gamma$ is \emph{$(\alpha,t)$-reachable} if the unique path of $\Gamma$ from $r_v$ to $x$ is $(\alpha,t)$-adapted. Let $R_{\alpha,t}(\omega)$ denote the set of $(\alpha,t)$-reachable vertices of $\Gamma$. Note that $R_{\alpha,t}(\omega)$ induces a connected subgraph of $\Gamma$ containing the root, and that (by K\"onig's infinity lemma) $R_{\alpha,t}(\omega)$ is infinite if and only if $\Gamma$ contains an infinite $(\alpha,t)$-adapted path. 

%For $p \in \{0,1\}$ and a tree $\Gamma$, let $\prob_{\Gamma,p}$ denote the probability measure on $\Omega(\Gamma)$ defined by 
We will be considering $p$-random subsets $X$ of $E(\Gamma)$. We first estimate, with an argument used in ([\ref{dz}], Proposition 2), the probability that a given path is $\alpha$-adapted with respect to $X$. Henceforth, we denote the `relative entropy' between $\alpha$ and $p$ by 
\[D(\alpha \Vert p) = \alpha \ln\left(\frac{\alpha}{p}\right) + (1-\alpha)\ln\left(\frac{1-\alpha}{1-p}\right).\] 
We remark that [\ref{dz}] defines $D(\alpha \Vert p)$ as the negative of this formula. \begin{lemma}\label{boundadapted} 
	Let $0 < p < \alpha < 1$. There exists $c > 0$ so that, if $[x_0,x_1,\dotsc, x_n]$ is a finite path, and $X$ is a $p$-random subset of the edges of the path, then  
	\[\prob\left(\ \!\![x_0, \dotsc, x_n] \textrm{ is $\alpha$-adapted w.r.t. $X$} \right)\ge cn^{-5/2}\exp(-n D(\alpha \Vert p)).\]
\end{lemma}
\begin{proof}
	We first make a claim that will simplify the estimate. 
	\begin{claim}	If $|X| \ge \alpha n$, then	there exists $\ell \in \{0, \dotsc, n-1\}$ such that the path corresponding to the cyclic ordering $[x_\ell,x_{\ell+1},\dotsc,x_n = x_0,x_1,\dotsc,x_{\ell}]$ is $\alpha$-adapted with respect to $X$.
	\end{claim} 
	\begin{proof}[Proof of claim:]
		For each $i \in \bZ_n$, let $t_i = 1-\alpha$ if the edge $x_{i}x_{i+1}$ is in $X$, and $t_i = -\alpha$ otherwise. For $0 \le j \le j' \le n$ let $S(j,j') = \sum_{i=j}^{j'-1} t_i$; observe that if $S(j,j') \ge 0$ then the path from $x_j$ to $x_{j'}$ has an $\alpha$-fraction of its edges in $X$. In particular, we have $S(0,n) = |X| - \alpha n \ge 0$. Choose $\ell \in \{0, \dotsc, n-1\}$ so that $S(0,\ell)$ is minimized. For $\ell \le h \le n$ we have $S(\ell,h) = S(0,h) - S(0,\ell) \ge 0$ and for $1 \le h \le \ell$ we have $S(\ell,n) + S(0,h) = S(0,n) + (S(0,h) - S(0,\ell)) \ge 0$. It follows from the observation that $\ell$ satisfies the claim. 
	\end{proof}
	By the above claim and symmetry, the probability that the path $[x_0, \dotsc, x_n]$ is $\alpha$-adapted is at least $\tfrac{1}{n}\prob (|X| \ge \alpha n)$. 
	%By Stirling's approximation, we have $\binom{n}{\lceil \alpha n \rceil} \sim n H\lb\tfrac{\lceil \alpha n \rceil}{n}\rb,$ where $H(x) = -x \ln x - (1-x) \ln (1-x)$. By continuity of $H$ and the fact that $\lceil \alpha n \rceil \sim \alpha n$, we therefore have $\binom{n}{\lceil \alpha n \rceil} \sim nH(\alpha)$. 
	
	 It is straightforward to show using $0 < \alpha < 1$ and Stirling's approximation that all sufficiently large $n$ satisfy 
	 \begin{align*} \lceil \alpha n \rceil! \le (\alpha n + 1)\lfloor \alpha n \rfloor! \le \sqrt{2\pi n^{3}} \left(\tfrac{\alpha n}{e}\right)^{\alpha n}\\
	 (n-\lceil \alpha n \rceil)! \le \sqrt{2\pi n}\left(\tfrac{(1-\alpha)n}{e}\right)^{(1-\alpha)n},
	 \end{align*}
	so Stirling's approximation gives $\binom{n}{\lceil \alpha n \rceil} \ge \tfrac{1}{\sqrt{2\pi}n^{3/2}}\left({\alpha^{\alpha}(1-\alpha)^{1-\alpha}}\right)^{-n}$ for all large $n$. All large enough $n$ thus satisfy
	\begin{align*}
		\tfrac{1}{n} \prob(|X| \ge \alpha n) &\ge \tfrac{1}{n}\prob(|X| = \lceil \alpha n \rceil) \\
		&= \frac{1}{n}\binom{n}{\lceil \alpha n \rceil}p^{\lceil \alpha n \rceil}(1-p)^{n-\lceil \alpha n \rceil}\\
		& \ge \frac{1}{\sqrt{2 \pi} n^{5/2}}\left(\frac{p^{\alpha}(1-p)^{1-\alpha}}{\alpha^\alpha(1-\alpha)^{1-\alpha}}\right)^np^{\lceil \alpha n \rceil - \alpha n}(1-p)^{\alpha n - \lceil \alpha n \rceil}\\
		& \ge \frac{p}{\sqrt{2\pi}n^{5/2}}\exp(-n D(\alpha \Vert p));
	\end{align*} 
	since the probability of a path being $\alpha$-adapted is clearly positive for all $n$, some $c \in  (0,\tfrac{p}{\sqrt{2\pi}}]$, obtained by taking a minimum over all small $n$, satisfies the lemma. %By Cram\'{e}r's large deviation formula [CITE SOMETHING], we have $\prob(|X| \ge \alpha n) \sim \exp(n D(\alpha \Vert p))$, so $\prob(|X| \ge \alpha n) \ge \tfrac{1}{2}\exp(nD(\alpha \Vert p))$ for large enough $n$, giving the lemma. 
	\end{proof}
We say a positive integer sequence $t = (t_i: i \ge 1)$ is \emph{slow} if it is nondecreasing and satisfies $\lim_{n \to \infty} t_n = \infty$ and $\lim_{n \to \infty} \tfrac{t_{n+1}}{\sum_{i = 1}^n t_i} = 0$.

The next lemma is the main technical result of this section. It shows that, if $t$ is a slow sequence, $G$ is a graph, and $\alpha$ and $p$ are chosen so that $\exp(D(\alpha \Vert p))$ is less than the graph invariant $\lambda_*(B(G))$ of the previous section, then there is some arc $e_0$ of $G$ for which a $p$-random subset of $E(\Gamma_{e_0}(G))$ will give an arbitrarily long $(\alpha,t)$-adapted path with probability bounded away from zero. The independence of $\delta$ on $n$ and $G$ in this lemma is crucial.

\begin{lemma}\label{uniformity}
	For all $0 < p < \alpha < 1$, every slow sequence $t$,  and all $\lambda > \exp(D(\alpha \Vert p))$, there is some $\delta = \delta(t,\lambda,\alpha,p) > 0$ such that, if $n \ge 1$ is an integer and $G$ is a connected graph of minimum degree at least $2$ with $\lambda_*(B(G)) \ge \lambda$, then there is an arc $e_0$ of $G$ so that, given a $p$-random subset $X \subseteq E(\Gamma_{e_0}(G))$, we have
	\[\prob\left(\text{$\Gamma_{e_0}(G)$ contains an $(\alpha,t)$-adapted path of length $n$ w.r.t. $X$}\right) > \delta.\]
\end{lemma}
\begin{proof}
	 Let $\lambda_* = \lambda_*(B(G))$. Let $t = (t_i: i \ge 1)$ and $T_\ell = \sum_{i=1}^\ell t_i$ for each $\ell \ge 0$. Let $\lambda_0 = \exp(D(\alpha\Vert p))$ and $\lambda_1,\lambda_2$ be real numbers so that $\lambda_0 < \lambda_1 < \lambda_2 < \lambda$. Note that $\lambda_0 > 1$ and $\lambda_* \ge \lambda$. 
	 
	  Let $\Pi(m)$ denote the probability that a path of length $m$ is $\alpha$-adapted with respect to a $p$-random subset of its edges, and for each $\ell \ge 0$ let $f(\ell) = \prod_{i=1}^{\ell}\Pi(t_i)^{-1}$ be the reciprocal of the probability that a path of length $T_{\ell}$ is $(\alpha,t)$-adapted.  To determine $\delta$, we first estimate $f$: 
	  	 
	 \begin{claim}
		There exists $M > 0$ such that $f(\ell+1) \le M \lambda_2^{T_\ell}$ for all $\ell$. 
	\end{claim}
	
	%(x/(1+e))^(1+d) < x 
	\begin{proof}[Proof of claim:]
		Let $c > 0 $ be given by Lemma~\ref{boundadapted} for $p$ and $\alpha$. We have 
		\begin{align*}
			f(\ell+1) = \prod_{i=1}^{\ell+1} \Pi(t_i)^{-1} &\le \prod_{i=1}^{\ell+1} \frac{t_i^{5/2}}{c} \exp\left(D(\alpha \Vert p) \sum_{j=1}^{\ell+1} t_j \right)\\	
			&= \lambda_0^{T_{\ell+1}}\prod_{i=1}^{\ell+1}\frac{t_i^{5/2}}{c}\\
			&= \lambda_1^{T_{\ell+1}}\prod_{i=1}^{\ell+1} \frac{t_i^{5/2}}{c}\left(\frac{\lambda_0}{\lambda_1}\right)^{t_i}\\
			&= \lambda_1^{(1+t_{\ell+1}/T_{\ell})T_{\ell}} \prod_{i=1}^{\ell+1}\frac{t_i^{5/2}}{c}\left(\frac{\lambda_0}{\lambda_1}\right)^{t_i}.
		\end{align*}
		Since $\lambda_0 < \lambda_1 < \lambda_2$ and $t_{\ell+1}/T_{\ell} \to 0$ and $t_{\ell} \to \infty$, this expression is at most $\lambda_2^{T_\ell}$ for large enough $\ell$. The claim follows by taking a maximum over all small $\ell$. 
	\end{proof}
	
	Set $\delta = M^{-1}(\tfrac{1}{\lambda_2}-\tfrac{1}{\lambda})$. Let $\bar{E}$ be the set of arcs of $G$, let $B = B(G)$ and let $w_*$ be the (strictly positive) eigenvector of $B$ for $\lambda_*$, normalised to have largest entry $1$. Choose $e_0 \in \bar{E}$ such that $w_*(e_0) = 1$. We show that $\delta$ and $e_0$ satisfy the lemma. 
	
	 For each $e \in \bar{E}$, let $b_{e}$ be the standard basis vector in $\bR^{\bar{E}}$ corresponding to $e$, and let $N_{h}(e_0,e) = b_{e_0}^T B^{h-1} b_{e}$ be the number of non-backtracking walks of length $h$ in $G$ with first arc $e_0$ and last arc $e$. 
	 
	 Let $\Gamma = \Gamma_{e_0}(G)$ and $r$ be the root of $\Gamma$. Let $\rho\colon V(\Gamma)\setminus \{r\} \to \bar{E}$ be the map assigning each walk to its last arc. Set $\phi(r) = 1$ and, for each vertex $x \ne r$ of $\Gamma$, set $\phi(x) = \lambda_*^{1-|x|}w_*(\rho(x))$. Note that $\phi((e_0)) = w_*(e_0) =  1$ and that, for each $x \ne r$ with $\rho(x) = e$, the sum of $\phi(y)$ over the children $y$ of $x$ is 
	\begin{align*}
		&\lambda_*^{1-(|x|+1)}\sum\left(w_*(e')\colon e' \in \bar{E}, B_{e,e'}=1\right)\\ 
		&= \lambda_*^{-|x|}b_{e}^TBw_* \\
		& = \lambda_*^{1-|x|}w_*(e) = \phi(x).\end{align*}
	(In other words, $\phi$ is a \emph{unit flow} on $\Gamma$.) It follows that for every $h \ge 0$ and all $x$ with $|x| \le h$, we have $\sum\left(\phi(y)\colon y \succeq x, |y| = h\right) = \phi(x)$. 
	
	For $X \subseteq E(\Gamma)$, we say that a vertex $v$ of $\Gamma$ is \emph{$(\alpha,t)$-reachable} with respect to $X$ if the path of $\Gamma$ from $r$ to $x$ is $(\alpha,t)$-adapted with respect to $X$; let $R(X)$ denote the set of $(\alpha,t)$-reachable vertices. Fix $\ell$ so that $T_\ell \ge n$, and define a random variable $Q = Q(X)$ by
	\[Q = f(\ell)\sum_{|x| = T_\ell} \phi(x)1_{R(X)}(x).\]
	The $\phi(x)$ sum to $1$ over all $x$ with $|x| =T_\ell$, so $\ev(Q) = 1$. We now bound the second moment of $Q$.
	\begin{claim}
		$\ev(Q^2) < \delta^{-1}$.
	\end{claim}
	\begin{proof}[Proof of claim:]
		We have \[\ev(Q^2) = f(\ell)^2 \sum_{|x|=|y| = T_\ell}\phi(x)\phi(y)\prob(x,y \in R(X)).\]
		For each $z \in V(\Gamma)$, let $k(z)$ be the maximum integer $k \ge 0$ so that $T_{k} \le |z|$. There are edge-disjoint paths of lengths $t_1,t_2,\dotsc,t_\ell$ and $t_{k(x \wedge y)+2},t_{k(x \wedge y)+3},\dotsc, t_\ell$ that all must be $\alpha$-adapted for both $x$ and $y$ to be in $R(X)$ (the first set of paths make up the path from $r$ to $x$ and the second set are contained in the path from $x \wedge y$ to $y$), so 
		\begin{align*}
			\prob(x,y \in R(X)) &\le \prod_{i=1}^{\ell} \Pi(t_i) \prod_{i=k(x \wedge y)+2}^\ell \Pi(t_i) \\
			&= f(k(x \wedge y)+1)f(\ell)^{-2}\\%\left(\prod_{i=1}^{\ell(x \wedge y)+1}\Pi(t_i)\right)^{-1}f(x)^{-1}f(y)^{-1}\\
			&\le M\lambda_2^{T_{k(x \wedge y)}} f(\ell)^{-2}\\
			&\le M\lambda_2^{|x \wedge y|}f(\ell)^{-2},
			%&\le w_*(e)^{-1}\prod_{i=1}^\ell t_i^3\exp(-t_i D(\alpha \Vert p)) f(x)^{-1}f(y)^{-1}, 
			\end{align*}
			where we use the first claim. Using the fact that $|x \wedge y| \ge 1$ whenever $|x| = |y| = T_{\ell}$, we have			
			\begin{align*}
				\ev(Q^2) &\le M\sum_{|x|,|y| = T_\ell}\phi(x)\phi(y)\lambda_2^{|x \wedge y|}\\
			&= M\sum_{1 \le |z| \le T_\ell}\lambda_2^{|z|}\sum_{\substack{|x| = |y| = T_\ell \\ x \wedge y = z}} \phi(x)\phi(y)\\
			&\le M \sum_{{1 \le |z| \le T_\ell}}\lambda_2^{|z|}\left(\sum_{\substack{|x| = T_\ell \\ x \succ z}}\phi(x)\right)^2\\
			&= M\sum_{\substack{1 \le |z| \le T_\ell}}\lambda_2^{|z|}\phi(z)^2\\				
			&= M\sum_{i=1}^{T_\ell} \lambda_2^i \sum_{|z| = i}\phi(z)^2.
			\end{align*}		
			If $|z| = i \ge 1$, then $w_*(e) \le 1$ gives \[\phi(z)^2 = \lambda_*^{2-2i}w_*(\rho(z))^2 \le \lambda_*^{2-2i} w_*(\rho(z)).\] For each $e \in \bar{E}$, the number of $z \in V(\Gamma)$ with $|z| = i$ and $\rho(z) = e$ is $N_{i}(e_0,e) = b_{e_0}^T B^{i-1} b_{e}$, so since $B w_* = \lambda_*w_*$ and $w_*(e_0) = 1$, we have
			\[\sum_{|z|=i} \phi(z)^2 \le \lambda_*^{2-2i} b_{e_0}^T B^{i-1}  \sum_{e \in \bar{E}} b_{e} w_*(e) = \lambda_*^{2-2i} b_{e_0}^TB^{i-1} w_* = \lambda_*^{1-i} \le \lambda^{1-i}.\]
			Thus $\ev(Q^2) < M\sum_{i=1}^{\infty}\lambda_2^i\lambda^{1-i} = M(\tfrac{1}{\lambda_2} - \tfrac{1}{\lambda})^{-1} = \delta^{-1}$ .
	\end{proof}
	 Now by the Cauchy-Schwartz inequality we have
	\[1 = \ev(Q)^2 = \ev(Q \cdot 1_{Q > 0})^2 \le \ev(Q^2) \ev(1_{Q>0}^2) < \delta^{-1} \prob(Q > 0),\]
	so $\prob(Q > 0) > \delta$. Therefore $\Gamma$ has an $(\alpha,t)$-adapted path of length $T_{\ell}$ with respect to $X$ with probability greater than $\delta$. Such a path contains an $(\alpha,t)$-adapted path of length $n$, giving the result. 	\end{proof}

\section{Graphs}

For a graph $G = (V,E)$ and for $p,\beta \in [0,1]$, let $f_p^{\beta}(G)$ denote the probability, given a $p$-random subset $X \subseteq E$, that $X$ contains at least a $\beta$-fraction of the edges of some circuit of $G$. Recall that $\lambda_*(\mu_0)$ is some value not less than $\mu_0-1$. 
%For each $x > 0$, let $\eps(x) = \tfrac{\eta(x)^3}{8x^3}$, where $\eta(x)$ as before denotes the distance from $x$ to the nearest integer. This is the term found in Lemma~\ref{boundlambda}; Note that $\eps(x) \ge 0$, with equality if and only if $x \in \bZ$.% A calculus argument shows that $x-1 + \eps(x)$ is a continuous increasing function for $x \ge 2$.

 %The next theorem, whose proof closely follows that of ([\ref{dz}], Theorem 4), would be slightly weakened if $\eps(\mu)$ were dropped from the statement; we prove the stronger result below (in which $\eps$ can be viewed as a small `penalty' term for nonintegral $\mu$) in order to obtain the equality characterisation in our main theorem. 

\begin{theorem}\label{maintech}
	For all $\mu_0 \ge 2$ and $0 < p < \beta < 1$ satisfying $\exp(D(\beta \Vert p)) < \lambda_*(\mu_0)$, there exists $\delta = \delta(\mu_0,p,\beta) > 0$ such that, if $G$ is a connected graph with $\mu(G) \ge \mu_0$, then $f_p^{\beta}(G) \ge \delta$.  
\end{theorem}
\begin{proof}
	It suffices to show this just for graphs of minimum degree at least $2$, since deleting a degree-$1$ vertex from a graph $G$ with $\mu(G) \ge 2$ does not change $f_p^{\beta}$ or connectedness, and does not decrease $\mu(G)$. Suppose that the result fails. Then there exists a sequence $G_1,G_2,\dotsc, $ of graphs of average degree at least $\mu_0$ and minimum degree at least $2$, such that $\lim_{n \to \infty}(f_p^{\beta}(G_n)) = 0$. We clearly have $f_p^{\beta}(G) \ge p^{d(G)}$ for every graph (this is the probability of a $p$-random subset containing \emph{every} edge in a given shortest cycle), so we may assume by taking a subsequence that $d(G_i) \ge i$ for each $i$. 
	
\begin{claim}
There is a slow integer sequence $t = (t_k\colon k \ge 1)$ so that $t_{|V(G_k)|} \le \sqrt{k}$ for each $k$.
\end{claim}
\begin{proof}[Proof of claim:]
Let $(t_k\colon k \ge 1)$ be a nondecreasing, divergent integer sequence in which  the integer $\lfloor \sqrt{r} \rfloor$ occurs at least $|V(G_r)|$ times for each $r \ge 1$. (Such a sequence can be chosen to diverge because each integer is only required to occur finitely often.) By construction we have $t_{|V(G_k)|} \le \lfloor \sqrt{k} \rfloor$ for each $k$. Furthermore, if $\ell \ge 1$ and $t_{\ell+1} = d+1 \ge 2$ then the integer $d$ has occured at least $|V(G_{d^2})| \ge d^2$ times before $t_{\ell+1}$, so $t_{\ell+1}/\sum_{i=1}^\ell t_i \le (d+1)/d^3$. It follows that $\lim_{n \to \infty} t_{n+1}/\sum_{i=1}^{n} t_i = 0$, so $(t_k\colon k \ge 1)$ is slow.
\end{proof}
	
	Note that $D(x \Vert p)$ is increasing in $x$ for $x > p$. Since $\exp(D(\beta \Vert p)) < \lambda_*(\mu_0)$ we can choose $\alpha \in (\beta,1)$ and $\lambda'$ so that 
	\[\exp(D(\beta \Vert p)) < \exp(D(\alpha \Vert p)) < \lambda' < \lambda_*(\mu_0).\]
	Let $k_0$ be large enough so that $\lambda_*(\mu_0;n) \ge \lambda'$ for all $n \ge k_0$. Let $\delta = \delta(t,\lambda',\alpha,p) > 0$ be given by Lemma~\ref{uniformity}. We argue that if $k$ is sufficiently large so that $k \ge k_0$ and $\tfrac{2\sqrt{k}+1}{k} \le \alpha - \beta$, then the graph $G = G_k$ satisfies $f_p^\beta(G) \ge \delta$. This contradicts $\lim_{n \to \infty} f_p^{\beta}(G_n) = 0$.

	%\end{proof}

%\begin{theorem}\label{maintech}
%	Let $\mu_0 \ge 2$ and let $(G_i: i \ge 1)$ be a sequence of connected graphs of minimum degree at least $2$ with strictly increasing girth and average degree at least $\mu_0$. If $0 < p < \beta < 1$ satisfy
%	\[\exp(-D(\beta \Vert p)) < \mu_0-1+\eps(\mu_0),\] then $\liminf_{n \to \infty} f_p^{\beta}(G_n) > 0$. 
%\end{theorem}
%\begin{proof}
	%Let $\mu_0 = \liminf_{n \to \infty} \mu(G_n)$, noting that $2 \le \mu \le \mu_0$ and so 
	%\[\exp(-D(\beta \Vert p)) < \mu-1+\eps(\mu) \le \mu_0-1+\eps(\mu_0).\] 
%	Since $d(G_n) \to \infty$, it is straightforward to see that there exists a slow integer sequence $t = (t_i: i \ge 1)$ so that $t_{|V(G_k)|} \le \sqrt{d(G_k)}$ for each $k \ge 1$. Set $\alpha \in (\beta,1)$ so that 
%	\[\exp(-D(\beta \Vert p)) < \exp(-D(\alpha \Vert p)) < \mu_0-1+\eps(\mu_0).\] 
	
%	Let $\delta = \delta(t, \mu_0-1+\eps(\mu_0),\alpha,p) > 0$ be given by Lemma~\ref{uniformity}. Let $G = G_i$ be a graph in the sequence of large enough girth so that $\alpha - \tfrac{2}{\sqrt{d(G)}} - \tfrac{1}{d(G)} \ge \beta$; we argue that $f_p^\beta(G) \ge \delta$ for all such $G$, which will imply the theorem. 
%for which $\mu(G)$ is sufficiently close to $\mu_0$ so that $\mu_0 \le \mu(G) \le 2\mu_0$ and $\eta(\mu(G)) \ge \tfrac{1}{2} \eta(\mu_0)$. Note that there are infinitely many choices for $G$; 	
	
	Let $G = G_k$ for such a $k$, and let $\Gamma = \Gamma_e(G)$ be the covering tree of $G$ with respect to the arc $e = (r,s)$ given by Lemma~\ref{uniformity}. Let $\pi\colon V(\Gamma) \to V(G)$ assign each path to its final vertex. Since $|V(G)| \ge k \ge k_0$, we have $\lambda_*(B(G)) \ge \lambda_*(\mu_0;k) \ge \lambda'$ by the choice of $k_0$.	
	
		  We now relate $f_p^\beta(G)$ to the probability that a $p$-random subset of $E(\Gamma)$ gives a long $(\alpha,t)$-adapted path. For each set $Z \subseteq V(G)$, let $G(Z)$ denote the subgraph of $G$ induced by $Z$. 
	
	 Recalling notation from the proof of Lemma~\ref{uniformity}, for $X \subseteq E(G)$ we say a vertex $v$ of $G$ is \emph{reachable} with respect to $X$ if $v = r$, or there is an $(\alpha,t)$-adapted path of $G$ (with respect to $X$) having first arc $e$ and last vertex $v$. We write $R(X)$ for the set of all such vertices.  Similarly, for $Y \subseteq E(\Gamma)$, we say a vertex $v$ of $\Gamma$ is \emph{reachable} with respect to $Y$ if there is an $(\alpha,t)$-adapted path of $\Gamma$ (with respect to $Y$) from the root to $v$. Let $R(Y)$ denote the set of all such vertices. Note, for any $X$ and $Y$, that each of the sets $R(X)$ and $\pi(R(Y))$ either is equal to $\{r\}$, or induces a connected subgraph of $G$ containing $r$ and $s$.
	 
	 	Suppose that $X$ is a $p$-random subset of $E(G)$ and $Y$ is a $p$-random subset of $E(\Gamma)$. Let $C_G$ denote the event that $G(R(X))$ contains a circuit, and $C_\Gamma$ denote the event that $G(\pi(R(Y)))$ contains a circuit.

	\begin{claim}
		$\prob(C_G) = \prob(C_{\Gamma})$. 
	\end{claim}
	\begin{proof}[Proof of claim:]
		Let $\cZ'$ denote the family of subsets of $V(G)$ that induce an \emph{acyclic} connected subgraph of $G$ containing $r$ and $s$, and let $\cZ = \cZ' \cup \{\{r\}\}$. The event $C_G$ fails to hold exactly when $R(X) \in \cZ$, so \[1-\prob(C_G) = \sum_{Z \in \cZ} \prob(R(X) = Z).\] Similarly, we have \[1-\prob(C_{\Gamma}) = \sum_{Z \in \cZ} \prob(\pi(R(Y)) = Z).\]
		If $Z = \{r\}$, then clearly $\prob(R(X) = Z) = \prob(\pi(R(Y)) = Z) = 1-p$. Suppose that $Z \in \cZ'$. By acyclicity of $G(Z)$, there is a unique subtree $\Gamma_Z$ of $\Gamma$ that contains the root of $\Gamma$ and satisfies $\pi(V(\Gamma_Z)) = Z$, and moreover $G(Z)$ and $\Gamma_Z$ are isomorphic finite trees. Now $G(Z)$ and $\Gamma_Z$ have the same number of edges, and the number of edges of $G$ with exactly one end in $Z \setminus \{r\}$ is equal to the number of edges of $\Gamma$ with exactly one end in $V(\Gamma_Z)$, so \[\prob(R(X) = Z) = \prob(R(Y) = V(\Gamma_Z)) = \prob(\pi(R(Y)) = Z).\] The claim now follows from the two summations above. 
	\end{proof}

	\begin{claim}
		$\prob(C_{\Gamma}) \ge \delta$. 
	\end{claim}
	\begin{proof}[Proof of claim:]
		By Lemma~\ref{uniformity}, the tree $\Gamma$ contains, with probability at least $\delta$, a length-$|V(G)|$ path $[v_1,v_2,\dotsc]$ that is $(\alpha,t)$-adapted with respect to $Y$. For any such path, there must be some $i < j$ so that $\pi(v_i) = \pi(v_j)$; now $\{\pi(v_i),\pi(v_{i+1}), \dotsc, \pi(v_j)\}$ is the vertex set of a closed non-backtracking walk of $G(\pi(R(Y)))$, which must contain a circuit. This implies the claim.
	\end{proof}
	\begin{claim}
		$f_p^{\beta}(G) \ge \prob(C_G)$.
	\end{claim}
	\begin{proof}[Proof of claim:]
		Suppose that $X \subseteq E$ satisfies $C_G$; i.e. $G(R(X))$ contains a circuit $C$. It suffices to show that $X$ contains a $\beta$-fraction of the edges of some circuit of $G$. Let $V(C) = [x_0,x_1,\dotsc,x_m]$, where $x_0$ is the end of a shortest $(\alpha,t)$-adapted path $P_0$ from $r$ to $V(C)$. If there is some $i \in\{1,\dotsc, m\}$ such that there exists in $G$ an $(\alpha,t)$-adapted path $P_i$ from $r$ to $x_i$ not containing $x_{i-1}$ and an $(\alpha,t)$-adapted path  $P_{i-1}$ from $r$ to $x_{i-1}$ not containing $x_i$, then $E(P_i) \cup E(P_{i-1}) \cup \{x_{i-1}x_i\}$ contains a circuit $C'$ of $G$. Moreover, this circuit is the disjoint union of the edge $x_{i-1}x_i$, a set of subpaths that are $\alpha$-adapted with respect to $X$, and at most two extra subpaths each of length at most $t_{|V(G)|}$ (these two subpaths are `partial' subpaths arising because the last intersection point of $P_{i-1}$ and $P_i$ need not cleanly divide these paths into a union of $\alpha$-dense subpaths), so $|X \cap E(C')| \ge \alpha|E(C')| - 2 t_{|V(G)|}-1$. Now $G = G_k$, so $|E(C')| \ge d(G) \ge k$ and $t_{|V(G)|} \le \sqrt{k}$, giving
		\[\tfrac{|X \cap E(C')|}{|E(C')|} \ge \alpha - \tfrac{2t_{|V(G)|}+1}{|E(C')|} \ge \alpha - \tfrac{2\sqrt{k} + 1}{k} \ge \beta,\]
		so $X$ contains a $\beta$-fraction of the edges of $C'$.
				
		If no such $i$ exists, then an easy inductive argument implies for each $j \ge 1$ that every $(\alpha,t)$-adapted path from $r$ to $x_j$ passes through $x_{j-1}$, so  $E(P_0)  \cup E(C) - \{x_0x_m\}$ is the edge set of an $(\alpha,t)$-adapted path from $r$ to $x_m$. By a similar argument to the above, we have $|E(C) \cap X| \ge \alpha |E(C)| - 2 t_{|V(G)|} - 1$, and thus $X$ contains a $\beta$-fraction of the edges of $C$, giving the claim.
	\end{proof}
	The last three claims give $f_p^{\beta}(G) \ge \delta$, implying the theorem. 	
\end{proof}

\section{The Threshold}

We now prove Theorems~\ref{maintechnical} and~\ref{mainintro}. Recall that, if $C$ is the cycle code of a graph $G$, then the probability of a maximum-likelihood decoding error in $C$ over a channel of bit-error rate $p \in (0,\tfrac{1}{2})$ is exactly the parameter $f_p^{1/2}(G)$ of the previous section. We use this fact to derive Theorem~\ref{maintechnical} (restated here) from Theorem~\ref{maintech}. 

\begin{theorem}\label{mainrev}
	If $R \in (0,1)$ and $\cG$ is the class of cycle codes of graphs, then $\theta_{\cG}(R) \le \tfrac{1}{2}\left(1 - \sqrt{1 - \tfrac{1}{\lambda^2}}\right)$, where $\lambda = \lambda_*(\tfrac{2}{1-R})$.  
\end{theorem}
\begin{proof}
	Fix $R \in (0,1)$, let $\mu = \tfrac{2}{1-R}$ and let $\theta = \tfrac{1}{2}\left(1 - \sqrt{1 - \tfrac{1}{\lambda^2}}\right)$, where $\lambda = \lambda_*(\mu)$. Note that $\exp(D(\half\Vert \theta)) = \lambda \ge \mu-1 > 1$ by Lemma~\ref{lambdastar}. It is enough to show that for all $p \in (\theta,\tfrac{1}{2})$ there is some $\eps > 0$ such that the probability of an error in maximum-likelihood decoding of a cycle code of rate at least $R$, over a binary symmetric channel with bit-error rate $p$, is at least $\eps$. 
	
	Let $p \in (\theta,\tfrac{1}{2})$. Since $p > \theta$ we have $\exp(D(\half \Vert p)) < \lambda$; let $\lambda_0 \in (\exp(D(\half \Vert p)),\lambda)$ and let $\mu_0 = \lambda_0 +1$. Let $\delta = \delta(\mu_0,p,\half)$ be given by Theorem~\ref{maintech} and set $\eps = \min(\delta,p^b)$, where $b = \tfrac{2\mu\mu_0}{\mu-\mu_0}$. 
	
	 Let $C$ be a cycle code of rate $R(C) \ge R$ and let $G$ be a connected graph whose cycle code is $C$. Note, since $R > 0$, that $G$ contains a circuit, so $f_p^{1/2}(G) \ge p^{|E(G)|}$. If $\mu(G) \ge \mu_0$ then $f_p^{1/2}(G) \ge \delta \ge \eps$ by Theorem~\ref{maintech}. Otherwise
	 \[1-\tfrac{2}{\mu} = R \le R(C) = 1 - \tfrac{2}{\mu(G)} + \tfrac{1}{|E(G)|} < 1 - \tfrac{2}{\mu_0} + \tfrac{1}{|E(G)|}, \]
	 so $|E(G)| < \tfrac{2\mu\mu_0}{\mu-\mu_0} = b$ and thus $f_p^{1/2}(G) \ge p^b \ge \eps$, as required. 
	\end{proof}
	
	Finally, we restate and prove Theorem~\ref{mainintro}.
	\begin{theorem}
	If $\cG$ is the class of cycle codes of graphs and $R \in (0,1)$, then $\theta_{\cG}(R) \le \tfrac{(1-\sqrt{R})^2}{2(1+R)}$. If equality holds, then $R = 1 - \tfrac{2}{d}$ for some $d \in \bZ$. 
	\end{theorem}
	\begin{proof}
		Let $\mu = \tfrac{2}{1-R}$ and $\lambda = \lambda_*(\mu)$. By Lemma~\ref{lambdastar} we have $\lambda \ge \mu-1$ with equality if and only if $\mu \in \bZ$. Theorem~\ref{mainrev} thus gives $\theta_{\cG}(R) \le \half\left(1-\sqrt{1 + \tfrac{2}{\mu-1}}\right)$, with equality only if $\mu \in \bZ$: that is, if and only if $R = 1 - \tfrac{2}{d}$ for some $d \in \bZ$. The result now follows from the definition of $\mu$ and a computation. 
	\end{proof}

\subsection*{Acknowledgements} We thank the two anonymous referees for their helpful suggestions that improved the quality of the paper.

\section*{References}

\newcounter{refs}

\begin{list}{[\arabic{refs}]}
{\usecounter{refs}\setlength{\leftmargin}{10mm}\setlength{\itemsep}{0mm}}

\item\label{ab06}
N. Alon and E. Bachmat, 
Regular graphs whose subgraphs tend to be acyclic,
Random Struct. Algo. 29 (2006), 324--337.

\item\label{alon}
N. Alon, S. Hoory and N. Linial,
The Moore Bound for Irregular Graphs,
Graph Combinator. 18 (2002), 53--57. 

\item\label{bmt78}
E.R. Berlekamp, R.J. McEliece and H.C.A. van Tilborg, 
On the inherent intractability of certain coding problems, 
IEEE Trans. Inform. Theory 24 (1978), 384--386.

\item\label{dz}
L. Decreusefond and G. Z\'{e}mor,
On the error-correcting capabilities of cycle codes of graphs,
Combin. Probab. Comput. 6 (1997), 27--38. 

\item\label{diestel}
R. Diestel, 
Graph Theory, 
Springer, 2000. 

\item\label{ggw}
J. Geelen, B. Gerards and G. Whittle, The highly connected matroids in minor-closed classes,
Ann. Comb. 19 (2015), 107--123. 

\item\label{gr}
C. Godsil and G. Royle, 
Algebraic Graph Theory,
Springer, 2001. 

\item\label{gelfand}
I. Gelfand,
Normierte ringe,
Rech. Math. [Mat. Sbornik] N.S., 9 (1941), 3--24

\item\label{lyons}
R. Lyons, 
Random walks and percolation on trees,
Ann. Probab. 18 (1990), 931--958.

\item\label{ms77} 
F. J. MacWilliams and N. J. A. Sloane, 
The Theory of Error-Correcting Codes,
Amsterdam, The Netherlands: North-Holland, 1977.

\item\label{nvz14}
P. Nelson and Stefan H.M. van Zwam, 
On the existence of asymptotically good linear codes in minor-closed classes,
IEEE Trans. Inform. Theory 61 (2015), 1153--1158. 

\item\label{nh81}
S.C. Ntafos and S.L. Hakimi, 
On the complexity of some coding problems, 
IEEE Trans. Inform. Theory 27 (1981), 794--796. 

\item\label{zt97}
J-P. Tillich, G. Z\'emor, 
Optimal cycle codes constructed from Ramanujan graphs, 
SIAM J. Discrete Math 10 (1997), 447--459. 
\end{list}

\end{document}